\documentclass[12pt]{article}
\usepackage{lcsys}
\usepackage{cite}
\usepackage{subcaption}
\usepackage{amsmath,amssymb,amsfonts}
\usepackage{algorithmic}
\usepackage{graphicx}
\usepackage{textcomp}
\allowdisplaybreaks
\usepackage{tikz}
\usepackage{hyperref}
\begin{document}
\title{Convergence of Decentralized Actor-Critic Algorithm in General-sum Markov Games}
\date{}
\author{Chinmay Maheshwari
\thanks{This version: November 2024. 
Corresponding author: chinmay\_maheshwari@berkeley.edu
} \thanks{Chinmay Maheshwari and Shankar Sastry are affiliated with Electrical Engineering and Computer Sciences (EECS) at UC Berkeley, Berkeley, CA, USA.}, Manxi Wu\thanks{Manxi Wu is affiliated with Operations Research and Information Engineering (ORIE) at Cornell University, Ithaca, NY, USA.}, and Shankar Sastry
\thanks{This work is supported by NSF Collaborative Research: Transferable, Hierarchical, Expressive, Optimal, Robust, Interpretable NETworks (THEORINET) under award No. DMS-2031899.}
}

\maketitle

\begin{abstract}
Markov games provide a powerful framework for modeling strategic multi-agent interactions in dynamic environments. 
Traditionally, convergence properties of decentralized learning algorithms in these settings have been established only for special cases, such as Markov zero-sum and potential games, which do not fully capture real-world interactions. 
In this paper, we address this gap by studying the asymptotic properties of learning algorithms in general-sum Markov games. In particular, we focus on a {decentralized algorithm where each agent adopts an actor-critic learning dynamic with asynchronous step sizes}. 
This decentralized approach enables agents to operate independently, without requiring knowledge of others' strategies or payoffs. 
We introduce the concept of a Markov Near-Potential Function (MNPF) and demonstrate that it serves as an approximate Lyapunov function for the policy updates in the decentralized learning dynamics, which allows us to {characterize the convergent set of strategies}. We further strengthen our result under specific regularity conditions and with finite Nash equilibria.
\end{abstract}

\newpage 
\section{Introduction}
The Markov game framework is a well-established model for capturing the coupled interactions of agents with heterogeneous utilities in dynamic and uncertain environments. This framework has been widely applied to various large-scale societal applications, including autonomous driving \cite{shalev2016safe}, smart grids \cite{ma2013scalable}, and e-commerce \cite{kutschinski2003learning}, among others. In such environments, agents must learn to adapt their behavior in the presence of other agents. Due to privacy and scalability concerns, these agents typically operate in a decentralized manner, meaning they act independently without direct communication or coordination with other agents and rely solely on local feedback they obtain from the environment. Furthermore, agents might not be aware of the existence of other agents in the environment. 

Existing literature has focused on design and analysis of decentralized learning algorithms in various settings that include the competitive setting represented as the two player zero-sum games (see \cite{sayin2021decentralized, daskalakis2020independent, wei2021last} and references therein), the cooperative setting represented by the Markov team games (see \cite{wheeler1986decentralized, arslan2016decentralized, yongacoglu2021decentralized} and references therein), and their generalizations to Markov potential games (see   \cite{maheshwari2022independent, fox2022independent} and references therein) or weakly acyclic games (see \cite{yongacoglu2023asynchronous} and references therein). However, these studies fail to capture the complexity of large-scale multi-agent interactions in the real world, which often involves a mixture of both cooperative and competitive dynamics. 
Recent research has explored decentralized learning in general-sum Markov games (e.g., \cite{mao2023provably, jin2021v, ma2023decentralized}), but are only concerned with convergence to weaker equilibrium concepts, such as the correlated equilibrium and coarse-correlated equilibrium, rather than reaching a Nash equilibrium.

This work investigates the convergence properties of decentralized actor-critic learning dynamics (proposed and studied in \cite{maheshwari2022independent} for Markov potential games) in general-sum Markov games, with a focus on evaluating the proximity of resulting policies to Nash equilibrium. It is well-known that decentralized learning dynamics often do not converge to Nash equilibria, even in static games \cite{hart2003uncoupled}. Therefore, our objective is to define the neighborhood around Nash equilibria where policies driven by decentralized actor-critic learning dynamics will persist in the long run in general-sum Markov games.

Our analysis builds on a new Markov near-potential function (MNPF) framework. In this framework, given any Markov game with finite state and action space, we can construct an MNPF such that the rate of the change in MNPF with respect to an agent's policy deviation approximates the rate of change in the agent's own value function (see Definition \ref{def: eps MPG new}), and the difference between the two rates is captured by the \emph{closeness parameter}. In the special case of the Markov potential game, which is heavily studied in literature, the MNPF becomes the exact Markov potential function, and the closeness parameter becomes zero.

Our idea of MNPF builds on the notion of near potential game in static games (introduced in \cite{candogan2013near, candogan2013dynamics}). In static games, a near potential function approximates the \emph{absolute} change of an agent's utility with their own strategy. The paper \cite{candogan2013dynamics} analyzed the convergence of fictitious play, showing that the convergent set of strategies is related to the closeness parameter of the near potential function.  This idea was extended to Markov games as the Markov $\alpha$-potential function
in \cite{guo2024alpha,guo2023markov, das2024learning}. {Our MNPF generalizes this by approximating the rate of change in agents' value functions with their policies, rather than absolute value changes ({Remark \ref{rem: ConnectMAPF}}). This generalization allows us to approximate the gradient of the value function (Lemma \ref{lem: GradDiff}), which is essential for characterizing the convergence set of the decentralized actor-critic algorithm (Remark \ref{rem: AdvantageMNPF})}. 
We show that MNPFs always exist for Markov games with finite state and action sets (Proposition \ref{prop:AllMNPG_Eq_Approximate}).


Our convergence result leverages the timescale separation in decentralized actor-critic dynamics, where agents update their local estimate of the Q-function on a faster timescale and their policies on a slower one, using only bandit feedback on state transitions and their own reward. Using two-timescale stochastic approximation theory \cite{perkins2013asynchronous}, we show that the fast updates converge to the Q-function's value while treating the slow policy updates as static. The agents' policy trajectories can then be analyzed as a continuous-time dynamical system, with the MNPF acting as an approximate Lyapunov function. We prove that these trajectories converge to a level set of the MNPF, which can be viewed as a set of approximate Nash equilibria (Theorem \ref{theorem:independent}). When the MNPF is Lipschitz continuous and the set of Nash equilibria is finite, the dynamics converge to the neighborhood of a single equilibrium (Theorem \ref{thm: ConvergenceFiniteEquilibrium}). In both theorems, the convergent set is characterized by the closeness parameter of the MNPG.   We evaluate the effectiveness of our results through a numerical experiment in Section \ref{sec: Numerics}. 

A schematic summarizing our approach is presented in Figure \ref{fig:my_diagram}. 

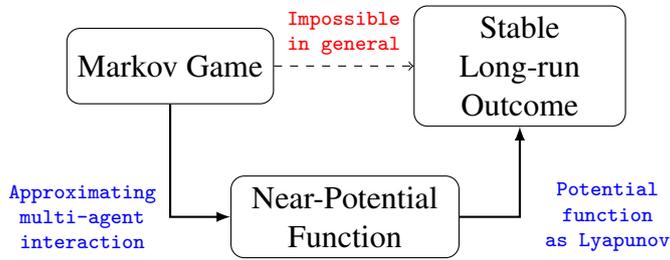
\begin{figure}[htbp] 
\centering
\begin{tikzpicture}

\node[draw, rectangle, minimum width=2cm, minimum height=1cm, rounded corners=5pt] (box1) at (0, 0) {Markov Game};
\node[draw, rectangle, minimum width=2cm, minimum height=1cm, rounded corners=5pt] (box2) at (4.6, 0) {\parbox{2.5cm}{\centering Stable Long-run\\Outcome}};
\node[draw, rectangle, minimum width=3cm, minimum height=1cm, align=center, rounded corners=5pt] (box3) at (2.3, -2) {\parbox{2.5cm}{\centering Near-Potential\\ Function}};

\draw[dashed, ->] (box1) -- node[above] {\color{red}\parbox{2.5cm}{\centering\scriptsize \texttt{Impossible} \\\texttt{in general}}} (box2);
\draw[thick, ->, >=latex] (box1) |- node[left] {\color{blue}\parbox{2cm}{\centering\scriptsize \texttt{Approximating} \\\texttt{multi-agent}\\\texttt{interaction}}} (box3);
\draw[thick, ->, >=latex] (box3) -| node[right] {\color{blue}\parbox{2cm}{\centering\scriptsize \texttt{Potential function} \\\texttt{as Lyapunov}}} (box2);

\end{tikzpicture}
\caption{Schematic of the our approach.}
    \label{fig:my_diagram}
\end{figure}

\textbf{Notations.} We denote the set of all probability distributions over a set \(X\) by \(\Delta(X)\). 
 For any function \(f:X\times Y\rightarrow \mathbb{R}\) we define \(f(x,p) = \mathbb{E}_{y\sim q}[f(x,y)]\), where \(p\in\Delta(X)\) and \(q\in \Delta(Y)\).  We use the notation \(\mathbf{1}_X\) to denote a vector of dimension \(|X|\) with all entries to be one. {We use \(\times_{i\in [N]}X_i\) to denote \(X_1\times X_2\times ... \times X_N.\) Unless otherwise stated, we use \(\|\cdot\|\) to mean $l_2$-norm.}

\section{Setup}
A (\emph{general-sum}) \emph{Markov game} \(\game\) is given by the tuple \(\langle \playerSet, \stateSet, (\actSet_i)_{i\in\playerSet}, (\stagePayoff_i)_{i\in\playerSet}, \transition, \discount\rangle\), where $\playerSet$ is a finite set of players (where \(|\playerSet|= N\)); \(\stateSet\) is a finite set of states; \(\actSet_i\) is a finite
   set of actions for each player \(i \in \playerSet\), with joint action profile \(a=(a_i)_{i \in I}  \in \actSet = \times_{i\in \playerSet}
   \actSet_i\); \(\stagePayoff_i: S \times A \to \mathbb{R}\) is the one-stage
   payoff function of player $i$ 
   and \(r_{\max} := \max_{i\in I, s\in S, a\in A}|r_i(s,a)|\); $\transition = (\transition(s'|s, a))_{s, s' \in S, a \in A}$ is the
   state transition matrix and $\transition(s'|s, a)$ is the probability that state changes from $s$ to $s'$ with action profile $a$; and \(\discount\in [0,1)\) is the discount factor.
 For each player \(i\), a \emph{stationary Markov policy} \(\pi_i: S \to \Delta(A_i)\) specifies the probability \(\pi_i(s, a_i)\) of choosing action \(a_i\) in state \(s\). The set of all stationary policies for player \(i\) is \(\Pi_i := \Delta(A_i)^{|S|}\), and the \emph{joint policy profile} of all players is \(\pi := (\pi_i)_{i\in\playerSet} \in \Pi := \times_{i \in I} \Pi_i\). Similarly, \(\pi_{-i} := (\pi_j)_{j \in\playerSet\setminus \{i\}}\) is the joint policy of all players except \(i\).
 
The game proceeds in discrete-time stages indexed by \(k \in \{0, 1, \dots\}\). At \(k=0\), the initial state \(s^0\) is sampled from a distribution \(\mu \in \Delta(S)\). At each time step \(k\), given state \(s^k \in S\), each player samples \(a^k_i \sim \pi_i(s^k)\), forming the joint action profile \(a^k := (a^k_i)_{i \in I}\). The next state is \(s^{k+1} \sim P(\cdot|s^k, a^k)\).
For an initial distribution \(\mu\) and a stationary policy profile \(\pi\), the expected total discounted payoff for each player \(i \in I\) is
\(V_i(\mu, \pi) = \mathbb{E}\left[\sum_{k=0}^{\infty} \gamma^k r_i(s^k, a^k)\right],\)
where \(s^0 \sim \mu\), \(a^k \sim \pi(s^k)\), and \(s^{k+1} \sim P(\cdot|s^k, a^k)\).
We define the discounted state occupancy measure as \(
    d^{\pi}_{\mu}(s) := (1-\gamma) \sum_{s^0 \in S} \mu(s^0) \sum_{k=0}^{\infty} \gamma^k \Pr(s^k = s | s^0),
\)
with \(\sum_{s \in S} d^{\pi}_{\mu}(s) = 1\).
    \begin{definition}[Stationary Nash equilibrium]\label{def: NashEqEpsNashEq} 
     For any $\epsilon \geq 0$, a policy profile \(\policyNash\in \Pi\) is an \emph{$\epsilon$-stationary Nash equilibrium} of $\game$ if for any \(i\in \playerSet\), any \(\policy_i\in\Pi_i\), 
     \(
        \vFunc_i(\mu, \policyNash_i,\policyNash_{-i}) \geq \vFunc_i(\mu, \policy_i,\policyNash_{-i}) -\epsilon.
    \)
    Any \(\epsilon\)-Nash equilibrium with \(\epsilon=0\) is a \emph{Nash equilibrium}. For any \(\epsilon\geq 0,\) we use the notation \(\textsf{NE}(\epsilon)\) to denote the set of all \(\epsilon\)-stationary Nash equilibria. 
    \end{definition}

\section{Markov Near-potential Function}
We introduce the notion of Markov near potential function which is crucial for subsequent disposition. 
\begin{definition}[Markov near-potential function ]\label{def: eps MPG new}
A bounded function \(\Phi: S \times \Pi \rightarrow \mathbb{R}\) is called a \emph{Markov near potential function} (MNPF) for a game \(\game\) with \emph{closeness parameter} \(\kappa \geq 0\), if for all \(s \in S\), \(i \in I\), \(\pi_i, \pi_i^{\prime} \in \Pi_i\), and \(\pi_{-i} \in \Pi_{-i}\), 
\begin{equation}\label{eqn: def eps MPG new}
\begin{aligned}
    &|\left(\Phi\left(s, \pi_i^{\prime}, \pi_{-i}\right)-\Phi\left(s, \pi_i, \pi_{-i}\right)\right) - \left(V_i\left(s, \pi_i^{\prime}, \pi_{-i}\right)-V_i\left(s, \pi_i, \pi_{-i}\right)\right)| \leq \nearParameter \|\pi_i'- \pi_i\|,
\end{aligned}
\end{equation} 
where the \(\|\pi_i'- \pi_i\|:= \sqrt{\sum_{s\in S}\sum_{a_i\in A_i}(\pi_i'(s, a_i)-\pi_i(s, a_i))^2}\).
\end{definition}
Definition \ref{def: eps MPG new} ensures that the difference between the rate of change in value function of any player with respect to the unilateral change in their policy and that of the potential function is upper bounded by $\kappa$.

\begin{remark}\label{rem: ConnectMAPF}{Our MNPF framework generalizes the recently introduced Markov \(\alpha\)-potential function framework from \cite{guo2023markov, guo2024alpha}. Specifically, if \({\Phi}\) is a Markov \(\alpha\)-potential function for a game \(G\), then for all \(s \in S\), \(i \in \playerSet\), \(\pi_i, \pi_i^{\prime} \in \Pi_i\), and \(\pi_{-i} \in \Pi_{-i}\),
\begin{align}\label{eq: MarkovAlphaPotentialFunction}
&|\left({\Phi}(s, \pi_i^{\prime}, \pi_{-i}) - {\Phi}(s, \pi_i, \pi_{-i})\right) \notag \\&\hspace{1cm}- \left(V_i(s, \pi_i^{\prime}, \pi_{-i}) - V_i(s, \pi_i, \pi_{-i})\right)| \leq \alpha,
\end{align}
which requires that the difference between an agent’s value function change and the Markov \(\alpha\)-potential function change due to a unilateral policy shift is \emph{uniformly bounded} by \(\alpha\). In contrast, our Definition \ref{def: eps MPG new} bounds this difference based on the magnitude of policy changes. Comparing \eqref{eqn: def eps MPG new} and \eqref{eq: MarkovAlphaPotentialFunction}, if \(\Phi\) is an MNPF for \(\game\) with parameter \(\nearParameter\), then \(\game\) is a Markov \(\alpha\)-potential game with \(\alpha \leq \kappa\sqrt{2|S|}\). 
}
\end{remark}

{
\begin{remark}\label{rem: AdvantageMNPF}
    Our Markov near-potential function framework \eqref{eqn: def eps MPG new} offers an advantage over the Markov \(\alpha\)-potential function \eqref{eq: MarkovAlphaPotentialFunction} by quantifying the gap between each agent's value function gradient and that of a potential function (Lemma \ref{lem: GradDiff}).
This property is essential for characterizing the convergent set of the decentralized actor-critic algorithm, as demonstrated in Theorems \ref{theorem:independent} and \ref{thm: ConvergenceFiniteEquilibrium}. To establish convergence, we show a positive correlation between policy update rates and the gradient of agents' value functions. Since this gradient closely approximates that of the potential function, with the closeness parameter defining the gap, we prove that the potential function consistently increases outside a neighborhood of the Nash equilibrium. This neighborhood, which defines the convergence set, varies in size based on the closeness parameter.
\end{remark}

The following proposition shows that a near-potential function can be constructed for any Markov game, and the approximate optimal solution of the near-potential function is an approximate Nash of the original game. 
{\begin{prop}\label{prop:AllMNPG_Eq_Approximate}
    For any game \(\game\),  there exists a tuple $(\Phi, \kappa)$ such that $\Phi$ is a MNPF of $\game$ with closeness parameter $\kappa$. Furthermore, for any \(\epsilon>0\) and any \(\pi^\ast\in \Pi\) such that \(\Phi(s,\pi^\ast) \geq \sup_{\pi\in \Pi} \Phi(s,\pi) - \epsilon\) for all \(s\in S\), \(\pi^\ast\) is a \((\nearParameter\sqrt{2|S|}+\epsilon)-\)stationary Nash equilibrium.
\end{prop}
}
\noindent\begin{proof}
For any game \(G\), we set 
    \(\Phi(s,\pi) = 0\) for every \(s\in S, \pi\in \Pi\). 
   We note that 
   {\small \begin{align*}
        &|V_i\left(s, \pi_i^{\prime}, \pi_{-i}\right)-V_i\left(s, \pi_i, \pi_{-i}\right)| \\
        \stackrel{\textit{(a)}}{\leq} &\frac{1}{1-\discount}\sum_{s'\in S,a_i\in A_i}\big| d^{\pi}_{\mu}(s')(\pi_i(s',a_i) - \pi_i'(s',a_i))\big|\cdot|{Q}_i(s',a_i;\policy')|\\ \stackrel{\textit{(b)}}{\leq} &\frac{r_{\max}}{(1-\gamma)^2} \|\pi_i-\pi_i'\|,
   \end{align*}}where \(\textit{(a)}\) is due to {multi-agent performance difference lemma (see Lemma \ref{lem: Multi-agentPerformanceDifference})}, and \(\textit{(b)}\) is due to the fact that for any \(\pi\in \Pi\),
    $$\max_{s,a_i}|Q_i(s,a_i;\pi)| \leq \frac{r_{\max}}{(1-\gamma)}$$ and 
    \(d_\mu^\pi(s)\in [0,1]\) for every \(s\in S\). Thus, \(\Phi\) is a MNPF for \(\game\) with the closeness-parameter \(\nearParameter = \frac{r_{\max}}{(1-\gamma)^2}\).

    {Next, we show that for any \(\epsilon>0\) and any \(\pi^\ast\in \Pi\) such that \(\Phi(s,\pi^\ast) \geq \sup_{\pi\in \Pi} \Phi(s,\pi) - \epsilon\) for all \(s\in S\), and thus \(\pi^\ast\) is a \((\nearParameter\sqrt{2|S|}+\epsilon)-\)stationary Nash equilibrium. Particularly, using \eqref{eqn: def eps MPG new}, we observe that for any \(i\in I,\pi_i'\in \Pi_i,\)
    \begin{align*}
        &V_i(\mu,\pi^\ast) - V_i(\mu,\pi_i',\pi_{-i}^\ast) \\
        \geq& \Phi_i(\mu,\pi^\ast) - \Phi_i(\mu,\pi_i',\pi_{-i}^\ast) - \kappa \|\pi_i'-\pi_i^\ast\| \geq -\epsilon -\kappa \sqrt{2|S|},
    \end{align*}
    where we note that \(\|\pi_i'-\pi_i^\ast\|\leq \sqrt{2|S|}\) using Cauchy Schwartz inequality. 
    }
\end{proof}

A  game \(\game\) may be associated with multiple MNPF with different $\kappa$. Proposition \ref{prop:AllMNPG_Eq_Approximate} suggests that an MNPF with a smaller closeness parameter \(\nearParameter\) is a better approximation of the original game in that the optimum of the potential function is a closer approximation of the Nash equilibrium in the original game.   
Following \cite{guo2023markov}, we can compute the MNPF with the smallest closeness parameter of a game as a semi-infinite linear program. We omit those details here for concise presentation.

\section{Decentralized Actor Critic Algorithm}
In this section, we present a decentralized learning algorithm in which each player makes decisions based solely on the current state information and their local reward feedback. Players do not need any knowledge about other players. Using MNPF, we provide theoretical guarantees on the long-run outcomes of the algorithm.

For every \(i\in I, s\in S, a_i\in A_i,\pi\in \Pi\), we define \emph{Q-function} as 
\begin{align}\label{eq: Q_i_function}
Q_i(s, a_i; \pi):= \stagePayoff_i(s,a_i, \pi_{-i}) + \discount \sum_{s'\in \stateSet} \transition(s'|s,a_i, \pi_{-i}) \vFunc_i(s', \pi),
\end{align}
    which is player $i$'s expected long-horizon discounted utility when the game starts in state \(s\) and they play action \(a_i\) in the 
    first stage and then employs policy \(\pi_i\) from the second stage onwards, and other players always employ policy $\pi_{-i}$.  
 With slight abuse of notation, we define \(Q_i(s;\pi) = (Q_i(s,a_i;\pi))_{a_i\in A_i}\in \R^{|A_i|}\). Furthermore, given \(Q_i\) and policy \(\pi\in \Pi\), we define the \emph{optimal one-stage deviation} of player $i$ in state \(s\in S\) as 
\begin{align}\label{eq: BestResponseExact}
    {\mathrm{br}}_i(s;\policy) &= \underset{\hat{\pi}_i\in \Delta(A_i)}{\arg\max}
    ~ \hat{\pi}_i^\top Q_i(s;\pi).
    \end{align}
\subsection{Decentralized Learning Algorithm and Preliminaries}\label{subsection:alg_preliminary}
We employ the discrete-time decentralized learning algorithm proposed in \cite[Algorithm 1]{maheshwari2022independent}, where each player adopts an actor-critic algorithm in a decentralized manner. In each iterate $t$ of the algorithm, every player $i \in I$ updates the following quantities: \textit{(i)} the counters \(\nk=(\nk(s))_{s \in S}\) and \(\nkak^t= (\nkak^t(s, \ai))_{s \in S, \ai \in \Ai}\), which keep track of the number of visits of all states and all state-action pairs up to the current iteration;
\textit{(ii)} their estimate of the local Q-functions \(q_i^t\), which is updated as a linear combination of the previous estimate and a new estimate based on the realized one-stage reward and the long-horizon discounted value from the next state as estimated from the q-function estimate and policy from previous iterate (refer \eqref{eq:d_q}); and \textit{(iii)} their local policies \(\piik\), which is updated as a linear combination of the policy in the previous iterate, and player $i$'s optimal one-stage deviation (refer \eqref{eq:d_pi}). 
Finally, every player samples an action $a_i^t\sim \pi_i^{t}$ with probability $(1-\theta)$ and from the uniform distribution over their action set $A_i$ with probability $\theta$, where \(\theta\in (0,1)\) is an exploration parameter.

\begin{remark}\label{rem: AsynchronousUpdate}
The updates in Algorithm \ref{alg:independent_decentralized} are \emph{asynchronous} because each player updates the counters, q-function estimate, and policy \emph{only} for the most recently visited state-(local) action pair, rather than for all state-action pairs.
\end{remark}

\begin{algorithm}[htp]
\textbf{Initialization:} $n^0(s)=0, \forall s \in\stateSet$; \(\tilde{n}_i^0(s, \ai)=0, \localQTilde_{i}^0(s, a_i)=0\), \(\pi_i^0(s)= 1/|\Ai|, ~\forall (i, \ai, s)\), and \(\explore\in (0,1)\). 
In stage 0, each player observes \(s^0\), choose their action $a_i^0 \sim \pi_i^0(s^0)$, and observe \(r_i^0=u_i(s^0,a^0)\). 

\textbf{In every iterate $t=1,2,...$,} each player observes $s^t$, and independently updates \(\{n^t_i,\nkak^t,\localQTilde_i^t,\policy_i^t\}\). \;

\smallskip
\textbf{Update $\nk,\nkak^t$}:  \(
    \nk(s^{t-1}) = n^{t-1}(s^{t-1})+1, \nkak^{t}(s^{t-1},a^{t-1}_i) = \nkak^{t-1}(s^{t-1},a^{t-1}_i)+1. 
\)

\textbf{Update $\localQTilde_i^t$}: Using the one-stage reward \(r_i^{t-1},\) update
\begin{equation}\label{eq:d_q}
\begin{aligned}
   & \localQTilde_{i}^{t}(s^{t-1},a^{t-1}_i) = \localQTilde_{i}^{t-1}(s^{t-1},a_i^{t-1})+ \stepFast(\nkak^{t}(s^{t-1},a_i^{t-1})) \\
    &\quad  \cdot\bigg( \reward_i^{t-1}+\discount 
   \pi_i^{t-1}(s^{t})^\top  \localQTilde_i^{t-1}(s^t)
    - \localQTilde_i^{t-1}(s^{t-1},a_i^{t-1})\bigg). 
\end{aligned}
\end{equation} 

\textbf{Update \(\policy_i^{t}\):} Pick \(\widehat{\br}_i \in \arg\max_{\pi_i\in\Delta(A_i)}\pi_i^\top q_i^{t-1}(s^{t-1})\) and set  
\begin{equation}\label{eq:d_pi}
\begin{aligned}
    \policy_i^{t}(s^{t-1}) =\policy_i^{t-1}(s^{t-1})+ \beta(\nk(s^{t-1}))\cdot(\widehat{\br}_i - \policy_i^{t-1}(s^{t-1})). 
\end{aligned}
\end{equation}

\textbf{Sample action and observe reward:} 
\begin{align}\label{eq: ActionSampling}
    a_i^t \sim (1-\explore)\pi_i^t(s^t)+ \explore \cdot (1/|A_i|)\mathbbm{1}_{A_i}.
\end{align}

Each player observes their own reward $r_i^t=u_i(s^t,a^t)$.
\caption{Decentralized Learning Algorithm}
\label{alg:independent_decentralized}
\end{algorithm}

{Next, we state assumptions that are central to analyze the convergence of Algorithm \ref{alg:independent_decentralized}. 
{
\begin{assumption}\label{as:basic}
(a) The initial state distribution $\mu(s)>0$ for all $s \in S$. 
Additionally, \(\min_{s,s'\in S, a\in A} P(s'|s,a) > 0\). 
\\
(b) The step-sizes satisfy \begin{itemize}
    \item[(i)] \(\sum_{n=0}^{\infty}\stepFast(n)=\infty, \sum_{n=0}^{\infty}\stepSlow(n)=\infty\), \(\lim_{n\rightarrow\infty}\stepFast(n)=\lim_{n\rightarrow\infty}\stepSlow(n)=0\);
    \item[(ii)] There exist some \(q,q'\geq 2\), \(\sum_{n=0}^{\infty}\stepFast(n)^{1+q/2}<\infty\) and \(\sum_{n=0}^{\infty}\stepSlow(n)^{1+q'/2}<\infty\);
    \item[(iii)] \(\sup_n \stepFast([xn])/\stepFast(n) < \infty\), \(\sup_n \stepSlow([xn])/\stepSlow(n) < \infty\) for all \(x\in (0,1)\), {where \([xn]\) denotes the largest integer less than or equal to \(xn\)}. Additionally, \(\{\stepFast(n)\},\{\stepSlow(n)\}\) are non-increasing in \(n\);
    \item[(iv)] \(\lim_{n\rightarrow\infty}\stepSlow(n)/\stepFast(n) = 0\). 
\end{itemize}
\end{assumption}
}
Assumption \ref{as:basic}-(a) is a standard assumption to ensure ergodicity of the Markov state transition for learning Q-functions. 
Additionally, Assumption \ref{as:basic}-(b) is standard assumption on step sizes in actor-critic algorithms \cite{maheshwari2022independent}. 
}

{

 To study the convergent set of Algorithm \ref{alg:independent_decentralized}, we apply two-timescale asynchronous stochastic approximation (TTASA) theory \cite{perkins2013asynchronous}. The discrete-time updates in Algorithm 1, along with Assumption \ref{as:basic}, satisfy the conditions in \cite{perkins2013asynchronous}, as noted in \cite{maheshwari2022independent}. TTASA theory ensures two things: First, the (fast) q-function estimates asymptotically track the Q-functions of the current policy. Using Assumption \ref{as:basic} and the contraction property of temporal difference operator, we observe that \[\lim_{t\to \infty} \|\localQTilde_{i}^{t}(s) - {Q}_i(s; \policy^{t}_i, \policy^{t,\theta}_{-i})\|_\infty=0\] holds with probability 1 for all \(s \in S\) and \(i \in \playerSet\), where
{\[
\pi_{-i}^{t,\theta}(s) := (1-\theta)\pi_{-i}^t+\theta(1/|A_i|)\cdot \mathbf{1}_{A_i}.
\]}
 Second, the convergent set of policy updates is the same as the convergent limit of any absolutely continuous trajectory of the following differential inclusion:
\begin{align}\label{eq: diff_inclusion}
    &\frac{d}{d\tau}{\varpi}_i^\tau(s) \in {\bar{\eta}(s)}\lr{ \br_i(s;\varpi^\tau_i, \varpi^{\tau,\theta}_{-i}) -\varpi_i^\tau(s)}, 
\end{align}
where $\tau\in [0, \infty)$ is a continuous-time index, 
{\[\varpi_{-i}^{\tau,\theta}(s):= (1-\theta)\varpi_{-i}^{\tau}(s)+\theta(1/|A_{-i}|)\cdot \mathbf{1}_{A_{-i}}, \quad  \forall \ s\in S, i \in \playerSet, \]}and $ \bar{\eta}(s)\in [\eta,1]$ for some positive scalar $\eta$ that depends on the ergodicity of the probability transition function. 
}

\subsection{Convergence Guarantees}
We now present the first main result of this paper, which characterizes the convergent set of policy updates in Algorithm \ref{alg:independent_decentralized} in terms of the superlevel set of a MNPF over the set of approximate Nash equilibria. This characterization is based on the closeness parameter 
\(\kappa\) associated with the MNPF.

\begin{theorem}\label{theorem:independent}
Consider a Markov game \(\game\) and an associated MNPF \(\Phi\) with closeness-parameter \(\nearParameter\). Under Assumptions \ref{as:basic}, 
 the sequence of policies $\{\pi^t\}_{t=0}^{\infty}$ induced by Algorithm \ref{alg:independent_decentralized}, with the exploration parameter {\begin{align}\label{eq: ExpParameter}\theta \leq {\lambda \sqrt{2|S|}}/({r_{\max} (2/(1-\gamma)^3 + 4/(D(1-\gamma)^2)) }),\end{align}} converge almost surely to the set 
 $$\Lambda := \{\pi: \Phi(\mu,\pi) \geq \min_{y\in \textsf{NE}(\Theta(\nearParameter+\lambda))}\Phi(\mu, y)\},$$ 
 where \(\lambda\) is a positive scalar, \[\Theta := DN^2 \sqrt{2|S|}/\eta, D = \frac{1}{1-\discount} \max_{i,\pi_{-i},s}\Big|{d^{\pi_i^\dagger, \pi_{-i}}_{\mu}(s)}/{\mu(s)}\Big|, \BR_i \in \arg\max_{\pi_i\in \Pi_i}{}{V}_i(\mu,\policy_i,\pi_{-i}),\]
  and $\eta$ is a positive scalar that depends on the ergodicity of the probability transition function.
\end{theorem}

\begin{proof}
Following Section \ref{subsection:alg_preliminary}, it is sufficient to show that every absolutely continuous trajectory of \eqref{eq: diff_inclusion} converges to the set \(\Lambda\).

We construct a Lyapunov function candidate $\phi: [0, \infty) \to \mathbb{R}$ as 
$$
\phi(\tau) = \max_{\varpi\in \Pi}{\Phi}(\mu,\varpi)- {\Phi}(\mu, \varpi^\tau),
$$
which is the difference of the MNPF at its maximizer with that of its value at \(\varpi^\tau\). 
The key step in the proof is to show that \(\phi(\tau)\) is weakly decreasing in \(\tau\) as long as \(\varpi^\tau\not\in \textsf{NE}(\Theta(\nearParameter+\lambda))\). 
We claim that it is sufficient to establish that for any \(\varpi^{\tau}\) that is an \(\epsilon-\)stationary Nash equilibrium, the following equation holds 
\begin{align}\label{eq: PhiDecay}
     \frac{d \phi(\tau)}{d\tau} \leq (\nearParameter+\lambda) N^2\sqrt{2|S|}-\frac{\eta}{D}\epsilon = \frac{\eta}{D} \Theta (\kappa+\lambda) - \frac{\eta}{D}\epsilon.
\end{align} 
Indeed, if \eqref{eq: PhiDecay} holds, then for any \(\epsilon > \Theta(\nearParameter+\lambda)\), \(\phi(\tau)\) decreases at a rate \(\eta/D\cdot(\epsilon-\Theta(\nearParameter+\lambda))\). Since \(\phi\) is bounded, any absolutely continuous trajectory of \eqref{eq: diff_inclusion} will enter the set \(\textsf{NE}(\Theta(\kappa+\lambda))\) in finite time, starting from any initial policy. Subsequently, even if trajectories leave this set, the function \(\Phi(\mu, \cdot)\) cannot decrease below \(\min_{\pi\in \textsf{NE}(\Theta(\nearParameter+\lambda))}\Phi(\mu,\pi)\), and once the trajectory leaves this set the potential function will always increase (as \(d\phi(\tau)/d\tau <0\)). Thus, it only remains to show  that \eqref{eq: PhiDecay} hold. Towards that goal, we note that 
\begin{align}\label{eq: dphi_firsteq}
    &\frac{d}{d\tau}\phi(\tau) = -\sum_{i\in I , s \in S}\lr{\frac{\partial {}{\potential}(\mu,\varpi^{\tau}) }{\partial \varpi_i(s)}}^\top\frac{d \varpi_i^{\tau}(s)}{d\tau} \notag\\ &
    = \sum_{i\in I, s \in S}\lr{ \frac{\partial {V}_i(\mu,\varpi^\tau) }{\partial \varpi_i(s)}  - \frac{\partial \Phi(\mu,\varpi^\tau) }{\partial \varpi_i(s)}}^\top\frac{d \varpi_i^\tau(s)}{d\tau}\notag \\ &\quad \quad -\sum_{i\in I, s\in S}\lr{\frac{\partial {V}_i(\mu,\varpi^\tau) }{\partial \varpi_i(s)}}^\top \frac{d \varpi_i^\tau(s)}{d\tau} \notag\\   
     &\stackrel{(i)}{\leq} \nearParameter \sum_{i\in I}\sqrt{\sum_{s\in S, a_i\in A_i}(\br_i(s,a_i;\varpi_i^\tau, \varpi_i^{\tau,\theta}) -\varpi_i^\tau(s, a_i))^2}
\notag  \\ &\quad \quad +  \sum_{i\in I, s\in S}\frac{d^{\varpi^\tau}_{\mu}(s)}{\discount-1}\bar{\eta}(s){Q}_i(s;\varpi^\tau)^\top\lr{ \widetilde{\mathrm{br}}_i^{\tau, \theta}(s)-\varpi_i^\tau(s)},\notag \\ 
&\leq \nearParameter N\sqrt{|S|\sum_{a_i\in A_i}\lr{\br_i(s,a_i;\varpi_i^\tau, \varpi_i^{\tau,\theta})+\varpi_i^\tau(s, a_i)}}
\notag  \\ &\quad \quad +  \sum_{i\in I, s\in S}\frac{d^{\varpi^\tau}_{\mu}(s)}{\discount-1}\bar{\eta}(s){Q}_i(s;\varpi^\tau)^\top\lr{ \widetilde{\mathrm{br}}_i^{\tau, \theta}(s)-\varpi_i^\tau(s)}\notag,
\\ 
&= \nearParameter N\sqrt{2|S|}
\notag  \\ &\quad \quad +  \sum_{i\in I, s\in S}\frac{d^{\varpi^\tau}_{\mu}(s)}{\discount-1}\bar{\eta}(s){Q}_i(s;\varpi^\tau)^\top\lr{ \widetilde{\mathrm{br}}_i^{\tau, \theta}(s)-\varpi_i^\tau(s)},
\end{align}
where \(\widetilde{\mathrm{br}}^{\tau, \theta}_i(s) \in {\mathrm{br}}_i(s;\varpi_i^\tau, \varpi^{\tau, \theta}_{-i})\), and \((i)\) is due to Lemma \ref{lem: GradDiff}. Next, by adding and subtracting the term \(\sum_{i\in I, s\in S}\frac{d^{\varpi^\tau}_{\mu}(s)}{\discount-1}\bar{\eta}(s)Q_i(s;\varpi_i^\tau, \varpi_{-i}^{\tau,\theta})^\top\lr{ \widetilde{\mathrm{br}}_i^{\tau, \theta}(s)-\varpi_i^\tau(s)}\) on the RHS in \eqref{eq: dphi_firsteq}, we obtain 
 \begin{align}\label{eq: EqrefTerm1and2}
\frac{d\phi(\tau)}{d\tau}&\leq \nearParameter N\sqrt{2|S|}
\notag  \\ &\quad \quad +  \underbrace{\sum_{i\in I, s\in S}\frac{d^{\varpi^\tau}_{\mu}(s)}{\discount-1}\bar{\eta}(s)({Q}_i(s;\varpi^\tau)-Q_i(s;\varpi_i^\tau, \varpi_{-i}^{\tau,\theta}))^\top\lr{ \widetilde{\mathrm{br}}_i^{\tau, \theta}(s)-\varpi_i^\tau(s)}}_{\texttt{Term 1}}\notag
\\ &\quad \quad +  \underbrace{\sum_{i\in I, s\in S}\frac{d^{\varpi^\tau}_{\mu}(s)}{\discount-1}\bar{\eta}(s)Q_i(s;\varpi_i^\tau, \varpi_{-i}^{\tau,\theta})^\top\lr{ \widetilde{\mathrm{br}}_i^{\tau, \theta}(s)-\varpi_i^\tau(s)}}_{\texttt{Term 2}},
    \end{align}
First, we bound \texttt{Term 1} in the above equation. Note that 
\begin{align}\label{eq: Term1}
    \texttt{Term 1} &\leq \frac{1}{1-\gamma}\sum_{i\in I}\max_{s\in S}\bigg|({Q}_i(s;\varpi^\tau)-Q_i(s;\varpi_i^\tau, \varpi_{-i}^{\tau,\theta}))^\top\lr{ \widetilde{\mathrm{br}}_i^{\tau, \theta}(s)-\varpi_i^\tau(s)}\bigg|\notag  \\ 
    &\leq \frac{2}{1-\gamma}\sum_{i\in I}\max_{s\in S, a_i\in A_i}|({Q}_i(s,a_i;\varpi^\tau)-Q_i(s,a_i;\varpi_i^\tau, \varpi_{-i}^{\tau,\theta}))|\notag
 \\
&\leq \frac{2\theta N^2}{(1-\gamma)^3}r_{\max},
\end{align}
where the last inequality is due to Lemma \ref{lem: Q_diff}. Next, we analyze \texttt{Term 2} in \eqref{eq: EqrefTerm1and2}.  Using \eqref{eq: BestResponseExact}, it holds that, for every \(i\in I\),
$${Q}_i(s;\varpi_i^\tau,\varpi_{-i}^{\tau, \theta})^\top\lr{ \widetilde{\mathrm{br}}_i^{\tau, \theta}(s)-\varpi_i^\tau(s)} \geq 0.$$
Furthermore, given the fact that \(\bar{\eta}(s)>\eta\) and \(d_{\mu}^{\varpi^\tau}(s)\geq 0\) for all \(i\in\playerSet,s\in\stateSet\), we bound  
\begin{align}\label{eq:two_stop}
    &\texttt{Term 2} \leq
  -\frac{\eta}{1-\discount}\sum_{i,s}d^{\varpi^\tau}_{\mu}(s){Q}_i(s;\varpi_i^\tau, \varpi_{-i}^{\tau,\theta})^\top \lr{ \widetilde{\mathrm{br}}_i^{\tau, \theta}(s)-\varpi_i^\tau(s)}.
    \end{align}
We claim that 
\begin{align}\label{eq: Claims}
&\sum_{i,s}d^{\varpi^\tau}_{\mu}(s){Q}_i(s;\varpi_i^\tau, \varpi_{-i}^{\tau,\theta})^\top \lr{ \widetilde{\mathrm{br}}_i^{\tau, \theta}(s)-\varpi_i^\tau(s)} \notag \\ &\geq 1/D\cdot\sum_{i,s}d^{\BR_i, \varpi_{-i}^{\tau,\theta}}_{\mu}(s){Q}_i(s;\varpi_i^\tau, \varpi_{-i}^{\tau,\theta})^\top \lr{ \BR_i(s)-\varpi_i^\tau(s)},
\end{align}
where \(\BR_i \in \arg\max_{\pi_i\in \Pi_i}{}{V}_i(\mu,\policy_i,\varpi_{-i}^\tau)\) be a best response\footnote{Note that $\BR_i$ maximizes the total payoff instead of just maximizing the payoff of one-stage deviation. Therefore, $\BR_i$ is different from the optimal one-stage deviation policy. We drop the dependence of $\BR_i$ on $\varpi_{-i}^\tau$ for notational simplicity. }  of player \(i\), given that the strategy of other players is \(\varpi_{-i}^\tau\). 
Indeed,   
note that 
\begin{align}
     &\sum_{i,s}d^{\BR_i, \varpi_{-i}^{\tau,\theta}}_{\mu}(s){Q}_i(s;\varpi_i^\tau, \varpi_{-i}^{\tau,\theta})^\top \lr{ \BR_i(s)-\varpi_i^\tau(s)}\notag \\ 
    &\leq \sum_{i,s}d^{\BR_i, \varpi_{-i}^{\tau,\theta}}_{\mu}(s)\max_{\hat{\pi}_i\in \Delta(A_i)}{Q}_i(s;\varpi_i^\tau, \varpi_{-i}^{\tau,\theta})^\top \lr{ \hat{\pi}_i(s)-\varpi_i^\tau(s)}\notag 
    \\
    &\leq \sum_{i,s}d_\mu^{\varpi^{\tau}}(s)\bigg\|\frac{d^{\BR_i, \varpi_{-i}^{\tau,\theta}}_{\mu}}{d_\mu^{\varpi^{\tau}}}\bigg\|_{\infty}\notag \cdot\max_{\hat{\pi}_i\in \Delta(A_i)}{Q}_i(s;\varpi_i^\tau, \varpi_{-i}^{\tau})^\top \lr{ \hat{\pi}_i(s)-\varpi_i^\tau(s)} \notag \\
    &\stackrel{\textit{(i)}}{\leq}   D \sum_{i,s}d_\mu^{\varpi^{\tau}}(s)\cdot\max_{\hat{\pi}_i\in \Delta(A_i)}{Q}_i(s;\varpi_i^\tau, \varpi_i^{\tau,\theta})^\top \lr{ \hat{\pi}_i(s)-\varpi_i^\tau(s)}
  \notag\\
    & = D \sum_{i,s}d_\mu^{\varpi^{\tau}}(s) \cdot{Q}_i(s;\varpi_i^\tau, \varpi_{-i}^{\tau,\theta})^\top \lr{ \widetilde{\mathrm{br}}_i^{\tau, \theta}(s)-\varpi_i^\tau(s)}, \label{eq:max}
\end{align}
where 
\textit{(i)} is due to the fact that \(d_{\mu}^{\varpi^{\tau,\theta}}(s)\geq (1-\discount)\mu(s)\) along with the definition of \(D\).

Using \eqref{eq: Claims} in \eqref{eq:two_stop}, we obtain 
  {\small \begin{align}\label{eq: PotentialFirst}
    &\texttt{Term 2} \notag \\ &\leq -\frac{\eta}{D(1-\discount)}\sum_{i,s}d^{\BR_i, \varpi^{\tau,\theta}_{-i}}_{\mu}(s)\cdot\max_{\hat{\pi}_i\in \Delta(A_i)}{Q}_i(s;\varpi_i^\tau, \varpi_{-i}^{\tau,\theta})^\top \lr{ \hat{\pi}_i(s)-\varpi_i^\tau(s)}\notag \\ 
    &\leq  -\frac{\eta}{D(1-\discount)}\sum_{i,s}d^{\BR_i, \varpi_{-i}^{\tau,\theta}}_{\mu}(s){Q}_i(s;\varpi_i^\tau, \varpi_{-i}^{\tau,\theta})^\top \lr{ \BR_i(s)-\varpi_i^\tau(s)},
\end{align}}
where the last inequality is because \(\pi_i^\dagger \in \Delta(A_i)\). 
Finally, note that 
\begin{align}\label{eq: PotentialSecond}
    &\sum_{i,s}d^{\BR_i, \varpi_{-i}^{\tau,\theta}}_{\mu}(s){Q}_i(s;\varpi_i^\tau, \varpi_{-i}^{\tau,\theta})^\top \lr{ \BR_i(s)-\varpi_i^\tau(s)}\notag \\ &= \sum_{i,s}d^{\BR_i, \varpi_{-i}^\tau}_{\mu}(s)\lr{Q_i(s;\varpi_i^\tau, \varpi_{-i}^{\tau,\theta})- V_i(s,\varpi_i^\tau, \varpi_{-i}^{\tau,\theta})}^\top(\pi_i^\dagger(s) - \varpi^\tau(s)) \notag 
\\
&= (1-\gamma)\sum_i{V}_i(\mu,\BR_i,\varpi_{-i}^{\tau,\theta})- {}{V}_i(\mu,\varpi_i^\tau, \varpi_{-i}^{\tau,\theta}),
\end{align}
where the last inequality is due to multi-agent performance difference lemma (Lemma \ref{lem: Multi-agentPerformanceDifference}). 
Combining \eqref{eq: PotentialFirst} and \eqref{eq: PotentialSecond}, we obtain 
\begin{align}\label{eq: Term2Final}
    \texttt{Term 2} 
    &{\leq}  -\frac{\eta}{D} \sum_{i}\lr{ {V}_i(\mu,\BR_i,\varpi_{-i}^{\tau, \theta})- {}{V}_i(\mu,\varpi_i^\tau, \varpi_{-i}^{\tau,\theta})} \notag \\ 
    &\leq -\frac{\eta}{D} \sum_{i}\lr{ {V}_i(\mu,\BR_i,\varpi_{-i}^{\tau})- {}{V}_i(\mu,\varpi_i^\tau, \varpi_{-i}^{\tau})}\notag  \\ &\quad + \frac{2\eta}{D}{\max_{\pi_i\in\Pi_i}\sum_i|V_i(\mu,\pi_i,\varpi_{-i}^{\tau,\theta})-V_i(\mu,\BR_i,\varpi_{-i}^{\tau})|} \notag  \\ 
    &\leq -\frac{\eta}{D} \sum_{i}\lr{ {V}_i(\mu,\BR_i,\varpi_{-i}^{\tau})- {}{V}_i(\mu,\varpi_i^\tau, \varpi_{-i}^{\tau})} + \frac{4\eta\theta N^2}{D(1-\gamma)^2}r_{\max},
\end{align}
where the last inequality is due to Lemma \ref{lem: V_diff}. Using \eqref{eq: Term1} and \eqref{eq: Term2Final} in \eqref{eq: EqrefTerm1and2}, we obtain 
\begin{align*}
    \frac{d\phi(\tau)}{d\tau}&\leq \nearParameter N\sqrt{2|S|} + \frac{2\theta N^2}{(1-\gamma)^3}r_{\max} + \frac{4\theta N^2}{D(1-\gamma)^2}r_{\max} \\ &\quad -\frac{\eta}{D} \sum_{i}\lr{ {V}_i(\mu,\BR_i,\varpi_{-i}^{\tau})- {}{V}_i(\mu,\varpi_i^\tau, \varpi_{-i}^{\tau})},
\end{align*}
where we used the fact that \(\eta\leq 1\). Since \(\theta \leq \lambda \frac{\sqrt{2|S|}}{r_{\max} (2/(1-\gamma)^3 + 4/(D(1-\gamma)^2)) }\), it ensures that 
\begin{align*}
     \frac{d\phi(\tau)}{d\tau} \leq (\kappa + \lambda )N^2\sqrt{2|S|} -\frac{\eta}{D} \sum_{i}\lr{ {V}_i(\mu,\BR_i,\varpi_{-i}^{\tau})- {}{V}_i(\mu,\varpi_i^\tau, \varpi_{-i}^{\tau})}. 
\end{align*}
From the definition of best response, \(V_i(\mu,\BR_i,\varpi_{-i}^{\tau}) \geq {V}_i(\mu,\pi_i,\varpi_{-i}^\tau)\) for every \(i\in I\). Furthermore, if \(\varpi^{\tau}\) is not an \(\epsilon-\)Nash equilibrium then there exists a player \(i\in I\) and a policy \(\pi_i\) such that  \(
{V}_i(\mu,\pi_i,\varpi_{-i}^\tau)- {}{V}_i(\mu,\varpi^\tau) \geq -\epsilon.
\)
Therefore, \(
     \frac{d \phi(\tau)}{d\tau} \leq (\nearParameter +\lambda)N^2\sqrt{2|S|}-\frac{\eta}{D} \epsilon. \)
This proves \eqref{eq: PhiDecay} and concludes the proof. 
\end{proof}

Next, we show that when game $\game$ has a finite number of equilibria and the function \(\Phi(\mu,\cdot)\) is Lipschitz (Assumption \ref{assm: FintelyManyEq}), Theorem \ref{theorem:independent} can be strengthened. Specifically, we show that the learning dynamics converge to neighborhood of equilibrium set instead of just converging to a level set of the MNPF. 
\begin{assumption}\label{assm: FintelyManyEq}
The equilibrium set is finite \(\textsf{NE}(0)=\{\pi^{\ast 1}, \pi^{\ast 2}, ..., \pi^{\ast K}\}\). 
\end{assumption}
Assumption \ref{assm: FintelyManyEq} has been adopted in previous literature in MARL (e.g. \cite{fox2022independent}). In fact, \cite{doraszelski2010theory} has shown that Assumption \ref{assm: FintelyManyEq} holds generically. For any $\delta\geq 0$, recall that $NE(\delta)$ is the set of approximate Nash equilibrium with approximation parameter $\delta$. We define $\Gamma(\delta)$ as the maximum distance between an approximate equilibrium in $NE(\delta)$ and the equilibrium set. 
\begin{align}\label{eq: f_def}
          \Gamma(\delta)  := \max_{\pi\in \textsf{NE}(\delta)} \min_{k\in[K]} \|\pi-\pi^{\ast k}\|.
\end{align}
{Since \(\textsf{NE}(0)\) denotes the set of Nash equilibria, \(\Gamma(0)=0\). }
\begin{theorem}\label{thm: ConvergenceFiniteEquilibrium}
Consider a Markov game \(\game\) and an associated MNPF \(\Phi\) with closeness-parameter \(\nearParameter\). Suppose Assumptions \ref{as:basic}-\ref{assm: FintelyManyEq} hold. There exists \(\bar{\nearParameter}, \bar{\epsilon}\)  such that if \(\nearParameter + \lambda \leq  \bar\nearParameter\), for some \(\lambda > 0\), the sequence of policies $\{\pi^t\}_{t=0}^{\infty}$ induced by Algorithm \ref{alg:independent_decentralized}, with the exploration parameter \eqref{eq: ExpParameter}, converge almost surely to the set \(
         \tilde{\Lambda}:= \left\{  \pi  | \exists \ k \in [K] : \ \|\pi-\pi^{\ast k}\| \leq \chi \right\},
\)
\begin{align}\label{eq: Chi_def}
\chi:=\min_{0\leq  \epsilon\leq \bar \epsilon}\Gamma(\epsilon+Z(\nearParameter+\lambda)) + ({2DLN\sqrt{|S|}\Gamma(Z(\nearParameter+\lambda))})/{(\eta\epsilon)},
\end{align} 
\(\lambda\) is a positive scalar, \(Z, D, \eta\) are defined in Theorem \ref{theorem:independent}, and \(L\) is defined in Lemma \ref{lem: LipschitzNearPF}. 
\end{theorem}
{
\begin{remark}
    In \eqref{eq: Chi_def}, \(\chi\) is weakly increasing with \(\kappa\) and \(\lambda\), as \(\Gamma(\cdot)\) is weakly increasing and upper-semicontinuous (see Lemma \ref{lem: Gamma_Prop}). If the game is a Markov potential game (i.e., \(\kappa = 0\)) and players have access to exact Q-functions, then they do not need to explore (i.e., \(\lambda = 0\)). In this case, the policy updates in \eqref{eq:d_pi}, evaluated at the exact Q-functions, converge to a Nash equilibrium. Indeed, with \(\kappa = 0\) and \(\lambda = 0\), one can follow the same steps as in the proof of Theorem \ref{thm: ConvergenceFiniteEquilibrium} to show that the second term in the optimization problem \eqref{eq: Chi_def} does not depend on \(\varepsilon\) since \(\Gamma(0) = 0\), and consequently, \(\chi = 0\).
\end{remark}
}

\begin{proof}[Proof of Theorem \ref{thm: ConvergenceFiniteEquilibrium}]
Similar to the proof of Theorem \ref{theorem:independent}, this result boils down to show that, under the conditions presented in Theorem statement, any absolutely continuous trajectory of the continuous time differential inclusion \eqref{eq: diff_inclusion} converges to the set \(\tilde{\Lambda}\). The proof comprises of two parts: first, we show that the trajectory of dynamical system \eqref{eq: diff_inclusion} will eventually converge to a close neighborhood of an approximate equilibrium, where one \(\pi^{\ast k}\) for some \(k\in [K]\). Second, we bound the maximum distance that the policy trajectory can be away from that equilibrium. These two steps together show that equilibrium will be refined in the set $\tilde{\Gamma}$ that is in the neighborhood of one equilibrium $\pi^{\ast k}$.

Before presenting each of these steps in detail, we present a few preliminary results that will be used later. Define 
$$
    d^\ast := \min_{k,l\in [K], k\neq l}\|\pi^{\ast k}- \pi^{\ast l}\|$$ 
as the minimum distance between a pair of equilibrium policies. Additionally, define $$\mc{B}(\pi;r):=\{\pi'\in \Pi: \|\pi'-\pi\|_2\leq r\}$$ as the neighborhood set of a policy $\pi$ with radius $r$. Let \(\zeta\) be such that \(\Gamma(\zeta)\leq  d^\ast/4\),  which exists due to Lemma \ref{lem: Gamma_Prop}. 
This construction ensures that for any \(k, l\in [K]\) such that \(k\neq l\) it holds that 
$$
    \mc{B}(\pi^{\ast k}; \Gamma(\zeta))\cap \mc{B}(\pi^{\ast l}; \Gamma(\zeta)) = \varnothing.
$$
Moreover, from the definition of \(\Gamma(\cdot)\) in \eqref{eq: f_def}, it must hold that for any \(\delta> 0\), \(
    \textsf{NE}(\delta)\subseteq \cup_{k\in [K]}\mc{B}(\pi^{*k}; \Gamma(\delta)).  
\)
Additionally, using the fact that for any \(k,l\in [K]\) such that \(k\neq l\), it must hold that \(\mc{B}(\pi^{\ast k}; \Gamma(\zeta))\cap \mc{B}(\pi^{\ast l}; \Gamma(\zeta)) = \varnothing,\). Therefore, we conclude that \(\textsf{NE}(\zeta)\) is contained in a disjoint union of sets, which is \(\sqcup_{k\in [K]}\mc{B}(\pi^{*k}; \Gamma(\zeta))\). We select \(\bar{\nearParameter}, \bar{\epsilon}\) such that \[\bar{\epsilon}+ \Theta\bar{\nearParameter} < \zeta/2, \text{and}, \Gamma(\bar{\epsilon}+\Theta\bar\nearParameter) \leq \frac{\zeta \eta d^\ast}{32NDL\sqrt S}.\] Lemma \ref{lem: Gamma_Prop} ensures that such construction exists. Finally, we define \(\kappa':=\nearParameter+\lambda\). 

From the proof of Theorem \ref{theorem:independent} we know that starting from any initial policy, any solution of \eqref{eq: diff_inclusion} eventually hits the set \(\textsf{NE}(\bar{\epsilon}+\Theta\nearParameter')\) in finite time.   Suppose that the trajectory leaves the component of the set \(\textsf{NE}(\bar{\epsilon}+\Theta\nearParameter')\) around the neighborhood of \(\pi^{\ast k}\) and enters the component in the neighborhood of \(\pi^{\ast l}\), for some \(k,l\in [K]\) such that \(k\neq l\). Since \(\bar{\epsilon}+\Theta\nearParameter'\leq \bar{\epsilon}+\Theta\bar{\nearParameter}<\zeta\), it holds that \(\textsf{NE}(\bar{\epsilon}+\Theta\nearParameter')\subseteq \textsf{NE}(\zeta),\) which is contained in the set \(\sqcup_{k\in [K]}\mc{B}(\pi^{\ast k}; \Gamma(\zeta))\). Therefore, any trajectory has to leave \(\mc{B}(\pi^{\ast k}; \Gamma(\zeta))\) and enter \(\mc{B}(\pi^{\ast l}; \Gamma(\zeta))\). Let \(t_1\) denote the time when the trajectory leaves the component of \(\textsf{NE}({\bar{\epsilon}+\Theta\nearParameter'})\) in neighborhood of \(\pi^{\ast k}\), \(t_2\) denote the time when it leaves \(\mathcal{B}(\pi^{\ast k}; {\Gamma(\zeta)})\), \(t_3\) denote the time when the trajectory enters \(\mathcal{B}(\pi^{\ast l}; {\Gamma(\zeta)})\), and \(t_4\) denotes the time when the trajectory enters the component of \(\textsf{NE}({\bar{\epsilon}+\Theta\nearParameter})\) around \(\pi^{\ast l}\). Since the function \(\phi(\tau)\) is decreasing outside \(\textsf{NE}({\bar\epsilon+\Theta\nearParameter'})\), it must hold that \( \Phi(\mu, \varpi^{t_2}) \geq \Phi(\mu, \varpi^{t_1})\) and \(\Phi(\mu, \varpi^{t_4}) \geq \Phi(\mu, \varpi^{t_3}).\)
 We claim that 
    \begin{align}\label{eq:pi_t_3_t_2_}
       \| \varpi^{t_2}-\varpi^{t_3}\| \geq d^\ast/2. 
    \end{align}
We show this by contradiction. Suppose that \(\|\varpi^{t_2} - \varpi^{t_3}\| < d^\ast/2\). Since \(\varpi^{t_2}\) lies on the boundary of \( \mathcal{B}(\pi^{\ast k}; {\Gamma(\zeta)})\) and \(\varpi^{t_3}\) on  \(\mathcal{B}(\pi^{\ast l}; {\Gamma(\zeta)})\), it must hold that 
$$
        \|\varpi^{t_2}-\pi^{\ast k}\| = \Gamma(\zeta) =\|\varpi^{t_3}-\pi^{\ast l}\|.
        $$
    Furthermore, using the definition of \(d^\ast\), we know that \(\|\pi^{\ast k} - \pi^{\ast l}\| \geq d^\ast\). But, from triangle inequality it must hold that 
    \begin{align*}
        \|\pi^{\ast k} - \pi^{\ast l}\| &\leq \|\pi^{\ast k} -\varpi^{t_2}\| + \| \varpi^{t_2}-\varpi^{t_3}\| + \|\varpi^{t_3} - \pi^{\ast l}\| \\&= 2\Gamma(\zeta) + \| \varpi^{t_2}-\varpi^{t_3}\|< d^\ast,
    \end{align*}
    which leads to a contradiction. Thus, \eqref{eq:pi_t_3_t_2_} holds. 

    Using \eqref{eq:pi_t_3_t_2_} and the fact that \(\|\dot{\varpi^\tau}\|\leq 2N\sqrt{|S|}\), it must hold that \(t_3-t_2 \geq d^\ast/(4N\sqrt{|S|})\). Additionally, using the fact the trajectory in the time interval \([t_2,t_3]\) lies outside \(\textsf{NE}({\zeta})\) and \eqref{eq: PhiDecay}, we conclude that 
    \begin{align}\label{eq: PhiT3MinusT2}
         \phi(t_3) - \phi(t_2) \leq - (\zeta - \Theta \bar{\nearParameter}) \frac{\eta d^\ast}{4ND\sqrt S}.
    \end{align}

Define 
$$\underline{\phi} := \min_{\pi\in \mc{B}(\pi^{\ast k}; {\Gamma(\bar{\epsilon}+\Theta\nearParameter' )})} \Phi(\mu, \pi), \quad  \bar{\phi} := \max_{\pi\in \mc{B}(\pi^{\ast l}; {\Gamma(\bar{\epsilon}+\Theta\nearParameter' ) })} \phi(\pi).$$
We claim that
\begin{align}\label{eq: PhiDiffEq}
\bar{\phi} \leq \underline{\phi}.     
\end{align}
Suppose that this is true, it shows that once the trajectory leaves the component of approximate equilibrium around \(\pi^{\ast k}\) and enters the component of \(\pi^{\ast l}\) it will never enter the component of approximate equilibrium around \(\pi^{\ast k}\)  in future. Thus, eventually the trajectories will visit the component of approximate equilibrium only around at most one equilibrium. Now we show \eqref{eq: PhiDiffEq}. Let \(y_k\in \mc{B}(\pi^{\ast k}; {\Gamma(\bar{\epsilon}+\Theta\nearParameter' )}), y_l\in \mc{B}(\pi^{\ast l}; {\Gamma(\bar{\epsilon}+\Theta\nearParameter' )})\) be such that \(\underline{\phi} =\phi(y_k),  \bar{\phi}=\phi(y_l)\).  This ensures that \(
    \|y_k - \varpi^{t_1}\| \leq 2 \Gamma(\bar{\epsilon}+\Theta\nearParameter')\) and \(\| y_l - \varpi^{t_4}\| \leq 2 \Gamma(\bar{\epsilon}+\Theta\nearParameter').
\)
Furthermore, from the Lipschitz property of potential function (see Lemma \ref{lem: LipschitzNearPF}), it holds that 
\begin{align*}
 \phi(\varpi^{t_2}) - \underline{\phi} \leq   \phi(\varpi^{t_1}) - \underline{\phi}   \leq L  \| y_k - \varpi^{t_1}\| \leq 2L\Gamma(\bar{\epsilon}+\Theta\bar \nearParameter) , \\ 
 \bar{\phi} - \phi(\varpi^{t_3}) \leq    \bar{\phi} - \phi(\varpi^{t_4}) \leq L  \| y_l - \varpi^{t_4} \| \leq 2L\Gamma(\bar{\epsilon}+\Theta\bar \nearParameter). 
\end{align*}
Combining the two inequalities and using \eqref{eq: PhiT3MinusT2}, we obtain 
$$\bar{\phi} - \underline{\phi} 
  \leq 4L\Gamma(\bar{\epsilon}+\Theta\bar \nearParameter)- (\zeta - \Theta \bar{\nearParameter}) \frac{\eta d^\ast}{4ND\sqrt S}.$$
The choice of \(\bar\kappa\) and \(\bar\epsilon\) ensures that \(4L\Gamma(\bar\epsilon+\Theta\bar\kappa)-  (\zeta - \Theta \bar{\nearParameter}) \frac{d^\ast\eta }{16NLS} < 0 \), which shows \eqref{eq: PhiDiffEq}.

To summarize, so far we have shown that  there exists a time \(T\) after which the trajectories will only visit the approximate equilibrium set that is close to one equilibrium. Let's say that the equilibrium is \(\pi^{\ast k}\). Next, we characterize the maximum distance the trajectory can travel around the equilibrium \(\pi^{\ast k}\). 

Fix \(\epsilon, \epsilon_1\) arbitrarily such that \(0\leq \epsilon_1 < \epsilon \leq \bar\epsilon\). 
Let \(\tau_1\) denote the time when the trajectory leaves the component of \(\textsf{NE}(\epsilon_1+\Theta\nearParameter')\) in neighborhood of \(\pi^{\ast k}\), \(\tau_2\) denote the time when the trajectory leaves \(\mc{B}(\pi^{\ast k}; \Gamma(\epsilon+\Theta\nearParameter'))\), \(\tau_3\) denote the time when the trajectory returns to this neighborhood again, and \(\tau_4\) denotes the time enters the component of \(\textsf{NE}({\epsilon_1+\Theta\nearParameter'})\) in neighborhood of \(\pi^{\ast k}\). 

Let \(r^\ast\) denote the maximum distance the trajectory goes outside of \(\mc{B}(\pi^{\ast k}; \Gamma(\epsilon+\Theta\nearParameter'))\). Since \(\|\dot{\varpi^\tau}\|\leq 2N\sqrt{|S|}\), it must hold that \(\tau_3-\tau_2 \geq \frac{r^\ast}{N\sqrt{|S|}}\). Furthermore, since \(\dot{\phi}\leq -\eta\epsilon/D\) during time interval \([\tau_2,\tau_3]\), it must hold that \(
    \phi(\tau_3) - \phi(\tau_2)\leq -\frac{r^\ast\eta\epsilon}{ND\sqrt{|S|}}. \)
Moreover, note that \(\phi(\tau_2)\leq \phi(\tau_1)\), and \(\phi(\tau_4)\leq \phi(\tau_3)\). This implies \(
    \phi(\tau_1) - \phi(\tau_4) \geq \phi(\tau_2) - \phi(\tau_3) \geq \frac{r^\ast\eta\epsilon}{ND\sqrt{|S|}}.\)
Furthermore, by Lipschitz continuity \(
   \phi(\tau_1) - \phi(\tau_4) \leq L\|\varpi^{\tau_1}-\varpi^{\tau_4}\| \leq L(\|\varpi^{\tau_1}-\pi^{\ast k}\| + \|\pi^{\ast k} - \varpi^{\tau_4}\|) \leq2L\Gamma(\epsilon_1+\Theta\nearParameter'). 
\)
Combining previous two inequalities we obtain that 
\begin{align*}
    \frac{r^\ast\eta\epsilon}{ND\sqrt{|S|}} \leq 2L\Gamma(\epsilon_1+\Theta\nearParameter') \implies r^\ast \leq \frac{2DLN\sqrt{|S|}\Gamma(\epsilon_1+\Theta\nearParameter')}{\eta\epsilon}. 
\end{align*}

The proof concludes by noting that the maximum distance the trajectories eventually goes away from \(\pi^{\ast k}\) is \(\Gamma(\epsilon+\Theta\nearParameter') + \frac{2DLN\sqrt{|S|}\Gamma(\epsilon_1+\Theta\nearParameter')}{\eta\epsilon}\). Since \(\epsilon, \epsilon_1\) are chosen arbitrarily, we conclude the proof of the theorem. 
\end{proof}

\section{Numerical Experiments}\label{sec: Numerics}

We demonstrate the performance of the proposed learning dynamics (Algorithm 1) in a perturbed Markov team game with a parallel link network of \(L\) links used by \(N\) travelers. At each stage \(k\), player \(i\) chooses a link \(a_i^k \in [L]\), and the network state is \(s = (s_\ell)_{\ell\in[L]}\), where \(s_\ell = 1\) means link \(\ell\) is unsafe. A link becomes unsafe with probability \(\nu_{1}\) if the number of users exceeds a threshold \(\thresh\), and \(\nu_{2}\) otherwise.

The utility of player \(i\) is a combination of individual and common rewards:
 \begin{align*}
     u_i(s,a)&= \underbrace{\rho_i\sum_{\ell=1}^{L}\mathbb{I}(a_i=\ell)\big(b_{\ell} - (1 + s_{\ell})m_{\ell}\sum_{j\in\playerSet}\mathbb{I}(a_{j} = \ell)\big)}_{\text{Individual Reward}}
     \\&+\underbrace{\sum_{i\in N}\sum_{\ell=1}^{L}\mathbb{I}(a_i=\ell)\big(b_{\ell} - (1 + s_{\ell})m_{\ell}\sum_{j\in\playerSet}\mathbb{I}(a_{j} = \ell)\big)}_{\text{Common Reward}}.
 \end{align*}
 Here, \(b_\ell\) is the fixed utility of using link \(\ell \in [L]\), \(m_\ell\) is a link-dependent constant which weights the effect of congestion on the cost and  $(1 + s_{\ell})m_{\ell}\sum_{j\in\playerSet}\mathbb{I}(a_{j} = \ell)$ represents the cost of using $\ell$ given the total number of users on $\ell$ and the network state. The utility depends on two terms: one is an individual reward function and another is the common reward term. The parameter \(\rho_i\) characterizes how much agent \(i\) weighs individual reward over the common interest.

Although the game is not a Markov potential game, the expected long-run common reward acts as a near-potential function with closeness parameter depending on \(\rho = \max_{i \in I}\rho_i\).

We simulate the system with \(N=4\), \(L=2\), \(\thresh=2\), \(\nu_1=0.8\), \(\nu_2=0.2\), \(m_1=2\), \(b_1=9\), \(m_2=4\), and \(b_2=16\), for \(5000\) stages. The step sizes are \(\alpha_i(n) = 1/n^{0.5}\) and \(\beta_i(n) = 1/n^{0.9}\). We set \(\rho_i = (i/N)\cdot K\), with \(K \in \{10,100,500,1000\}\). Initial link states are random. 

For the exploration parameter \(\theta=0.05\) and the perturbation parameter \(\rho_i\) with \(K=10\), policies converge close to Nash equilibrium (Figure \ref{fig:ind_decentralized}). Interestingly, we observe that the Nash gap goes close to zero even though \(K\neq 0\) indicating that the parameter \(\kappa\) for this game is a very small number. Furthermore, we see that the Nash approximation gaps increase with \(\theta\) and \(K\) (Figure \ref{fig:enter-label}), which validates our theoretical convergence guarantees.

\begin{figure*}[!ht]
        \centering
    \begin{minipage}[l]{\textwidth}
        \centering
    \begin{subfigure}{0.32\textwidth}
    \centering
\includegraphics[width=\textwidth]{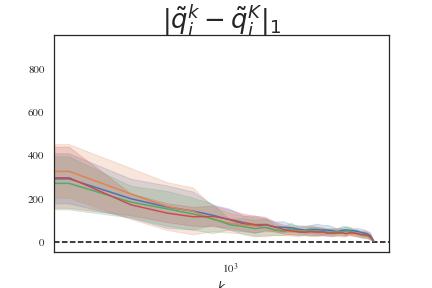}
    \caption{}
    \label{fig: F1}
    \end{subfigure}
    \begin{subfigure}{0.32\textwidth}
    \centering 
\includegraphics[width=\textwidth]{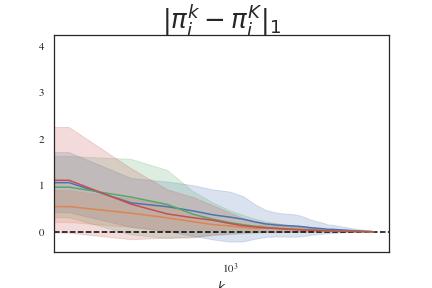}
    \caption{}
    \label{fig: F2}
    \end{subfigure}
    \begin{subfigure}{0.32\textwidth}
    \centering
\includegraphics[width=\textwidth]{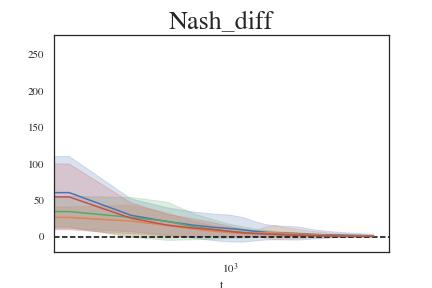}
    \caption{}
    \label{fig: F3}
    \end{subfigure}
     \end{minipage}
    \caption{Convergence of q-estimate, policies, and Nash error after \(10^5\) steps of Algorithm 1. In each of the figures the four curves correspond to four players. 
    Each curve represents the mean value of the quantity over \(5\) trials, and we give error margins of \(\pm 1\) standard deviation.}
    \label{fig:ind_decentralized}
\end{figure*}

\begin{figure}[h!]
    \centering
    \includegraphics[width=\linewidth]{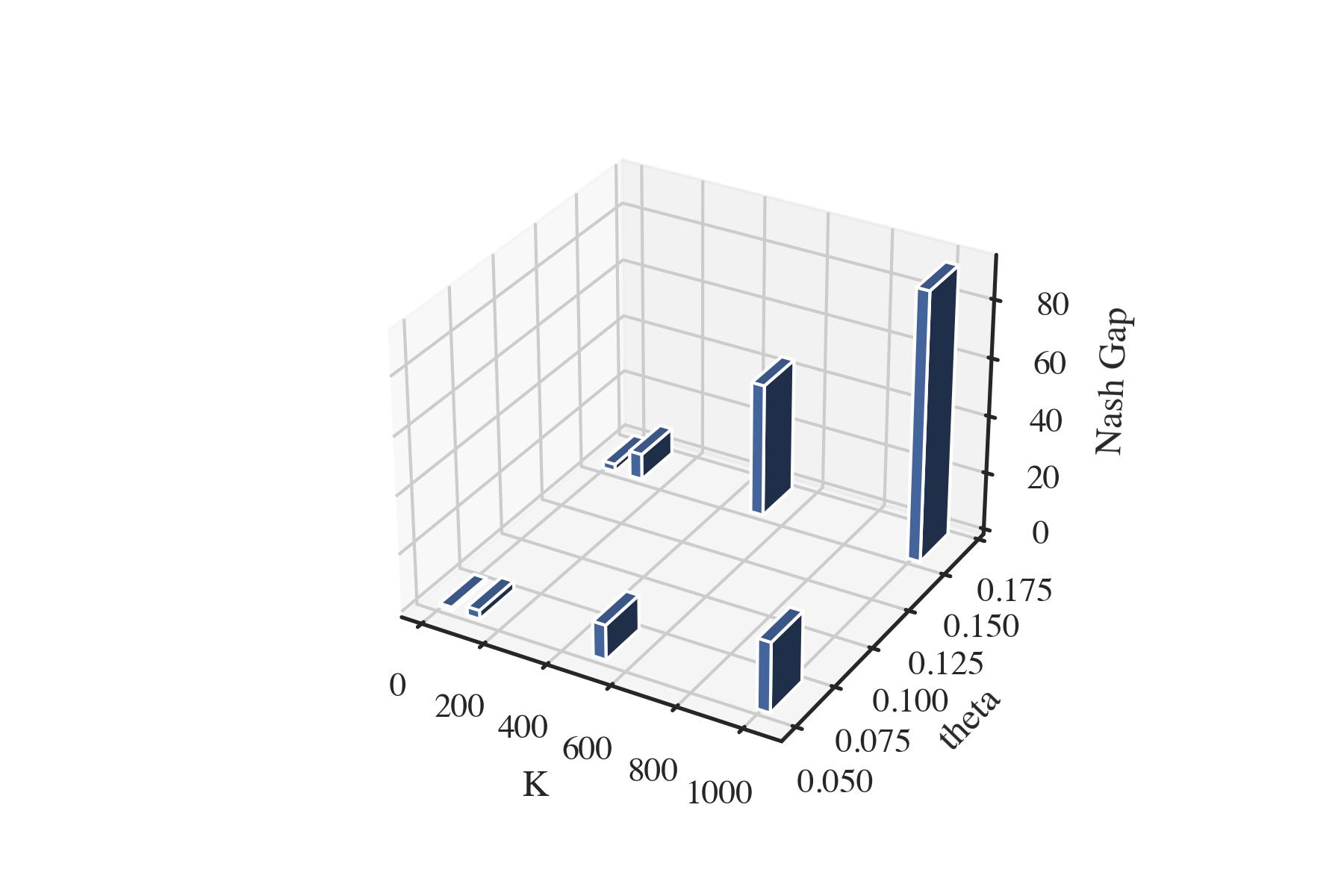}
    \caption{Variation of Nash gap with change in the exploration rate \(\theta\) and the reward perturbation \(K\). Increasing both of these parameters increases the Nash gap, which is defined to be \(\max_{i\in \playerSet}\max_{s\in S}|V_i(s,\pi^{t_{\max}}) - \max_{\pi_i\in \Pi_i}V_i(s,\pi_i,\pi_{-i}^{t_{\max}}) |\), where \(t_{\max}\) is the number of iterations of Algorithm 1.}
    \label{fig:enter-label}
\end{figure}
\section{Conclusion}
We analyze the convergence of a decentralized actor-critic algorithm in general Markov games by constructing a Markov near-potential function (MNPF) that approximates changes in value functions due to individual policy updates, effectively serving as an approximate Lyapunov function. {The MNPF framework offers potential for designing new algorithms in Markov games beyond the decentralized actor-critic analysis.}

\bibliography{refs}
\bibliographystyle{ieeetr}
\newpage 
\appendix 
\section{Auxiliary Results}
We present a technical result used in proof of Theorem \ref{theorem:independent}. 
\begin{lemma}\label{lem: GradDiff}
    For any Markov game \(G\) and an associated MNPF \(\Phi\) with closeness parameter \(\nearParameter\), it holds that 
    \begin{align*}
     | v_i^\top \frac{\partial \Phi(\mu,\pi)}{\partial \pi_i} -  v_i^\top \frac{\partial V_i(\mu,\pi)}{\partial \pi_i} | \leq \nearParameter\|v_i\|_2, \quad \forall i\in I, v_i\in \mathbb{R}^{|S||A_i|}. 
 \end{align*}
\end{lemma}
\begin{proof}
For any \(v_i\in \mathbb{R}^{|S|\cdot|A_i|}\) it holds that 
\begin{equation}\label{eq: DirDer}
 \begin{aligned}
     v_i^\top \frac{\partial \Phi(\mu,\pi)}{\partial \pi_i} &= \lim_{h\rightarrow 0} \frac{\Phi(\mu,\pi_i+hv_i,\pi_{-i})-\Phi(\mu,\pi_i,\pi_{-i})}{h}, \\ 
     v_i^\top \frac{\partial V_i(\mu,\pi)}{\partial \pi_i} &= \lim_{h\rightarrow 0} \frac{V_i(\mu,\pi_i+hv_i,\pi_{-i})-V_i(\mu,\pi_i,\pi_{-i})}{h}.
 \end{aligned}
 \end{equation}
Additionally, using the definition of MNPF, we obtain
 \begin{equation}\label{eq: DiffVal}
 \begin{aligned}
    V_i(\mu,\pi_i+hv_i,\pi_{-i})-V_i(\mu,\pi_i,\pi_{-i}) - \nearParameter h\|v_i\|_2\\ \leq \Phi(\mu,\pi_i+hv_i,\pi_{-i})-\Phi(\mu,\pi_i,\pi_{-i}) \\ \leq V_i(\mu,\pi_i+hv_i,\pi_{-i})-V_i(\mu,\pi_i,\pi_{-i}) + \nearParameter h\|v_i\|_2.
 \end{aligned}
 \end{equation}
Dividing everything by \(h\) in \eqref{eq: DiffVal} and using \eqref{eq: DirDer}, we obtain 
 \begin{align*}
     | v_i^\top \frac{\partial \Phi(\mu,\pi)}{\partial \pi_i} -  v_i^\top \frac{\partial V_i(\mu,\pi)}{\partial \pi_i} | \leq \nearParameter\|v_i\|_2, \quad \forall i\in I, v_i\in \mathbb{R}^{|S|\cdot|A_i|}.
 \end{align*}
This concludes the proof. 
\end{proof}

We now present a technical result used in proof of Theorem \ref{thm: ConvergenceFiniteEquilibrium}.

\begin{lemma}\label{lem: LipschitzNearPF}
    For any Markov game \(G\), an associated MNPF \(\Phi\) with closeness parameter \(\nearParameter\), and \(\mu\in \Delta(S),\) the mapping \(\pi\mapsto\Phi(\mu,\pi)\) is \(L-\)Lipschitz continuous, with
    \begin{align*}
        L= \left(\nearParameter\sqrt{N}+\frac{r_{\max}}{(1-\gamma)^2}\sqrt{N|S||A|}\right).
    \end{align*}
\end{lemma}
\begin{proof}
To show the desired Lipschitz bound, it is sufficient to show that 
\begin{align*}
    |\Phi(\mu,\pi)-\Phi(\mu,\pi')| \leq L \|\pi-\pi'\|, \quad \forall \ \pi, \pi'\in \Pi. 
\end{align*}
For the remaining proof, consider two arbitrary policies \(\pi = (\pi_1, \pi_2,...,\pi_N),\) and \(\pi' = (\pi_1', \pi_2', ..., \pi_N')\). For any \(i\in \{0,1,..,N\}\), define a joint policy 
\[\pi^{(i)} = (\pi_1, \pi_2, \cdots, \pi_{i}, \pi_{i+1}', \cdots, \pi'_N).\] Naturally, \(\pi^{(0)}=  \pi'\) and \(\pi^{(N)} = \pi\). 

Note that 
\begin{align*}
    |\Phi(\mu,\pi)-\Phi(\mu,\pi')| = |\Phi(\mu,\pi^{(N)})-\Phi(\mu,\pi^{(0)})| \leq  \sum_{i=0}^{N-1}|\Phi(\mu,\pi^{(i+1)})-\Phi(\mu,\pi^{(i)})|. 
\end{align*}

Since \(\pi^{(i+1)}\) and \(\pi^{(i)}\) only differ in the policy of player \((i+1)\), using Definition \ref{def: eps MPG new}, we obtain 
\begin{align}\label{eq: Eq1}
    &|\Phi(s,\pi)-\Phi(s,\pi')| \leq \kappa\sum_{i=0}^{N-1}\|\pi^{(i+1)}-\pi^{(i)}\|+\sum_{i=0}^{N-1}|V_{i+1}(s,\pi^{(i+1)})-V_{i+1}(s,\pi^{(i)})| \notag   \\ 
    &\leq \kappa\sum_{i=0}^{N-1}\|\pi_{i+1}-\pi_{i+1}'\|\notag \\&\quad +\sum_{i=0}^{N-1}\Bigg|\sum_{s'\in S, a_{i+1}'\in A_{i+1}}\frac{\partial V_{i+1}(s,\tilde \pi)}{\partial \tilde\pi_{i+1}(s',a_{i+1}')} \bigg|_{\tilde{\pi}=\zeta^{(i)}}(\pi_{i+1}(s',a_{i+1}')-\pi_{i+1}'(s',a_{i+1}'))\Bigg|,
\end{align}
where \(\zeta^{(i)}= (\pi_1, \pi_2, \cdots, \pi_i,\xi_{i+1}',\pi'_{i+2}...,\pi_N')\) and \(\xi_{i+1}' = \pi_{i+1} + t(\pi_{i+1}-\pi_{i+1}')\) for some \(t\in [0,1]\).

Using \eqref{eq: Eq1}, along with Lemma \ref{lem: V-gradient}, we obtain 
\begin{align*}
    &|\Phi(\mu,\pi)-\Phi(\mu,\pi')|\\
    &\leq \kappa\sum_{i=0}^{N-1}\|\pi_{i+1}-\pi_{i+1}'\| \\&\quad + \frac{1}{(1-\gamma)}\sum_{i=0}^{N-1}\Bigg|\sum_{s'\in S, a_{i+1}'\in A_{i+1}}d^{\zeta^{(i)}}_{\mu}(s')Q_{i+1}(s',a_{i+1}',\zeta^{(i)}) (\pi_{i+1}(s',a_{i+1}')-\pi_{i+1}'(s',a_{i+1}'))\Bigg|,
    \\
    &\leq \kappa\sum_{i=0}^{N-1}\|\pi_{i+1}-\pi_{i+1}'\| \\&\quad + \frac{1}{(1-\gamma)}\sum_{i=0}^{N-1}\sum_{s'\in S, a_{i+1}'\in A_{i+1}}\Bigg|d^{\zeta^{(i)}}_{\mu}(s')Q_{i+1}(s',a_{i+1}',\zeta^{(i)}) (\pi_{i+1}(s',a_{i+1}')-\pi_{i+1}'(s',a_{i+1}'))\Bigg|.
\end{align*}

Additionally, using the fact that \(|d_{\mu}^{(\zeta^{(i)})}(s')|\in [0,1]\), we obtain 
\begin{align*}
     &|\Phi(\mu,\pi)-\Phi(\mu,\pi')|
      \leq \kappa\sum_{i=0}^{N-1}\|\pi_{i+1}-\pi_{i+1}'\| \\&\quad + \frac{1}{(1-\gamma)}\sum_{i=0}^{N-1}\sum_{s'\in S, a_{i+1}'\in A_{i+1}}\Bigg|Q_{i+1}(s',a_{i+1}',\zeta^{(i)}) (\pi_{i+1}(s',a_{i+1}')-\pi_{i+1}'(s',a_{i+1}'))\Bigg| \\ 
      &\leq \kappa\sum_{i=0}^{N-1}\|\pi_{i+1}-\pi_{i+1}'\| \\&\quad + \frac{1}{(1-\gamma)}\sum_{i=0}^{N-1}\max_{s',a_{i+1}',\zeta^{(i)}}Q_{i+1}(s',a_{i+1}',\zeta^{(i)}) \sum_{s'\in S, a_{i+1}'\in A_{i+1}}\Bigg|(\pi_{i+1}(s',a_{i+1}')-\pi_{i+1}'(s',a_{i+1}'))\Bigg|.
\end{align*}

From \eqref{eq: Q_i_function},  we note that
\begin{align*}
\max_{s',a_{i+1}',\zeta^{(i)}}|Q_{i+1}(s',a_{i+1}',\zeta^{(i)})|  \leq \frac{r_{\max}}{(1-\gamma)}. 
\end{align*}
where \(r_{\max} = \max_{i,s,a} r_i(s,a) \). Therefore, we obtain, 
\begin{align*}
     &|\Phi(\mu,\pi)-\Phi(\mu,\pi')|
      \\
      &\leq \kappa\sum_{i=0}^{N-1}\|\pi_{i+1}-\pi_{i+1}'\| \\&\quad + \frac{u_{\max}}{(1-\gamma)^2}\sum_{i=0}^{N-1} \sum_{s'\in S, a_{i+1}'\in A_{i+1}}\Bigg|(\pi_{i+1}(s',a_{i+1}')-\pi_{i+1}'(s',a_{i+1}'))\Bigg|.
\end{align*}
Finally, by using Cauchy-Schwarz inequality, we conclude that 
\begin{align*}
    &|\Phi(s,\pi)-\Phi(s,\pi')|
     \\ &\leq \left(\kappa \sqrt{N} + \frac{u_{\max}}{(1-\gamma)^2}\sqrt{|S||A|N}\right)\|\pi-\pi'\|. 
\end{align*}

\end{proof}
\begin{lemma}\label{lem: Gamma_Prop}
    For any \(\delta\geq 0\), the function \(\Gamma(\delta)\) (defined in \eqref{eq: f_def}) exists, is upper semi-continuous and weakly increasing.  
\end{lemma}
\begin{proof}
First, we show that \(\Gamma(\cdot)\) exists for every \(\delta\geq 0\) and is upper semicontinous. This follows directly from the fact that \(\textsf{NE}(\delta)\) is a non-empty and compact for any non-negative \(\delta\), and the mapping \(\pi\mapsto \min_{k\in [K]}\|\pi-\pi^{\ast k}\|\) is continuous. Furthermore, the upper semicontinuity follows directly from Berge's maximum theorem.  Finally, we show that \(\Gamma(\cdot)\) is a weakly increasing function. This is due to the fact that \(\textsf{NE}(\delta)\subseteq \textsf{NE}(\delta')\) for any \(\delta' > \delta\geq 0\).  
\end{proof}

\begin{lemma}[Lemma 3.8(a) in \cite{maheshwari2022independent}]\label{lem: V-gradient}
    For any \(\mu\in\Delta(S), s'\in S,\pi\in\Pi,i\in [N],a_i'\in A_i\),
\begin{align*}
    \frac{\partial {V}_i(\mu,\tilde{\pi})}{\partial\tilde{\pi}_i(s',a_i')}\bigg|_{\tilde{\pi}=\pi}&=\frac{1}{1-\gamma}d^{\pi}_{\mu}(s'){Q}_i({s}',a_i';\pi),
\end{align*}
where 
\begin{align}
   d^{\pi}_{\mu}(s') :=  (1-\gamma)\sum_{k=0}^{\infty}\gamma^k\mathsf{Pr}(s^k=s'|s^0\sim\mu) 
\end{align} is a probability distribution over the set of states. Here, \(s^0\sim\mu, s^{k+1}\sim P(\cdot|s^{k},a^{k})\) and \(a^k\sim\pi(s^k)\). 
\end{lemma}

\begin{lemma}[Lemma 3.8 (c) in \cite{maheshwari2022independent}]\label{lem: V_diff}
    For any \(i\in I, \pi_i\in \Pi_i, \pi_{-i} \in\Pi_{-i}\) it holds that 
    \begin{align*}
        \max_{s\in S }|V_i(s,\pi_i,\pi_{-i}) - V_i(s,\pi_i,\pi_{-i}^\theta)| \leq \frac{2\theta |I|}{(1-\gamma)^2}r_{\max} ,
    \end{align*}
    where \(r_{\max} := \max_{i,s,a}|r_i(s,a)|\), \(\pi_{-i}^\theta(s) := (1-\theta)\pi_{-i}(s) + \theta\pi^\circ\) and \(\pi^\circ := (1/|A_{-i}|)\mathbbm{1}_{A_{-i}}\). 
\end{lemma}

\begin{lemma}[Lemma 3.8 (d) in \cite{maheshwari2022independent}]\label{lem: Q_diff}
    For any \(i\in I, \pi_i\in \Pi_i, \pi_{-i}\in \Pi_{-i}\), it holds that 
    \begin{align*}
        \max_{s,a_i}|Q_i(s,a_i;\pi_i,\pi_{-i}) -Q_i(s,a_i;\pi_i,\pi_{-i}^\theta)  | \leq \frac{2\theta |I|}{(1-\gamma)^2}r_{\max}, 
    \end{align*}
    where \(r_{\max} := \max_{i,s,a}|r_i(s,a)|\), \(\pi_{-i}^\theta(s) := (1-\theta)\pi_{-i}(s) + \theta\pi^\circ\) and \(\pi^\circ := (1/|A_{-i}|)\mathbbm{1}_{A_{-i}}\). 
\end{lemma}

{
\begin{lemma}[\cite{maheshwari2022independent, ding2022independent}]\label{lem: Multi-agentPerformanceDifference}
For any policy \(\policy=(\policy_i,\policy_{-i}), \policy'=(\policy_i',\policy_{-i})\in \Pi\) and any \(\mu\in\Delta(\stateSet)\), 
\begin{align*}
&{\vFunc}_i(\mu,\policy)-{\vFunc}_i(\mu,\policy') = (1/(1-\discount))\cdot\sum_{s'}d^{\pi}_{\mu}(s') \advFunc_i(s',\policy_i;\policy'),
\end{align*}
where $\advFunc_i(s,a_i;\policy) \defas {Q}_i(s,a_i;\policy)-{V}_i(s,\policy)$.
\end{lemma}}

\end{document}